\documentclass[11pt,letterpaper]{article}

\usepackage{amsmath,amssymb,amsthm}
\usepackage{booktabs} 
\usepackage{xspace}
\usepackage[pagebackref=true,hidelinks]{hyperref}
\usepackage{accents} 
\usepackage{pgf} 
\usepackage{xparse} 
\usepackage{bbm}
\usepackage{enumitem}
\usepackage[numbers]{natbib}
\usepackage{algorithm}
\usepackage[noend]{algorithmic}
\usepackage{verbatim}
\usepackage[margin=1in]{geometry}

\usepackage{color}
\newcommand{\kgnote}[1]{{\color{orange}{#1}}}

\newcommand{\nsi}[1]{{\color{red}{#1}}}

\newtheorem{lemma}{Lemma}
\newtheorem{theorem}{Theorem}
\newtheorem{claimnum}{Claim}

\theoremstyle{definition}
\newtheorem{definition}{Definition}
\newtheorem{example}{Example}

\usepackage{todonotes}

\usepackage{accsupp}

\newtheorem{corollary}[lemma]{Corollary}

\theoremstyle{definition}

\newcommand{\R}{\mathbb{R}}
\newcommand{\Z}{\mathbb{Z}}

\newcommand{\argmax}{\operatorname{arg\,max}}

\newcommand{\opt}{{\normalfont\textsc{OPT}}}

\newcommand{\E}{\mathbb{E}}

\def\R{\ensuremath{\mathbb{R}}}

\def\opt{\ensuremath{\textsc{opt}}}

\def\Pr{\ensuremath{\mathrm{Pr}}}

\def\pf{\ensuremath{\underline{p}}}
\def\pc{\ensuremath{\overline{p}}}

\def\med{\ensuremath{\mathrm{med}}}
\def\Cmed{\ensuremath{C_{\med}}}

\defcitealias{CDW}{CDW16}
\defcitealias{FGKK}{FGKK16}
\defcitealias{DW}{DW17}
\defcitealias{DHP}{DHP17}
\defcitealias{Myerson}{Myerson81}
\defcitealias{DDT15}{DDT15}
\defcitealias{SSW17}{SSW17}

\begin{document}
\title{Reducing Inefficiency in Carbon Auctions \\with Imperfect Competition}

\author{Kira Goldner%
\thanks{%
    {Columbia University (\url{kgoldner@cs.columbia.edu}).  Supported in part by NSF CCF-1420381 and by a Microsoft Research PhD Fellowship.  Supported in part by NSF award DMS-1903037 and a Columbia Data Science Institute postdoctoral fellowship.}}
\and Nicole Immorlica%
\thanks{%
    {Microsoft Research (\url{nicimm@microsoft.com}).  Supported in part by NSF Award 1841550, ``EAGER: Algorithmic Approaches for Developing Markets.''}}
\and Brendan Lucier%
\thanks{%
        {Microsoft Research (\url{brlucier@microsoft.com})}} }

\date{\today}
\maketitle

\begin{abstract}
We study auctions for carbon licenses, a policy tool used to control the social cost of pollution.  Each identical license grants the right to produce a unit of pollution.  Each buyer (i.e., firm that pollutes during the manufacturing process) enjoys a decreasing marginal value for licenses, but society suffers an increasing marginal cost for each license distributed.  The seller (i.e., the government) can choose a number of licenses to put up for auction, and wishes to maximize the societal welfare: the total economic value of the buyers minus the social cost.
%
Motivated by emission license markets deployed in practice, we focus on uniform price auctions with a price floor and/or price ceiling.  The seller has distributional information about the market, and their goal is to tune the auction parameters to maximize expected welfare.  The target benchmark is the maximum expected welfare achievable by any such auction under truth-telling behavior.  Unfortunately, the uniform price auction is not truthful, and strategic behavior can significantly reduce (even below zero) the welfare of a given auction configuration.

We describe a subclass of ``safe-price'' auctions for which the welfare at any Bayes-Nash equilibrium will approximate the welfare under truth-telling behavior.  We then show that the better of a safe-price auction, or a truthful auction that allocates licenses to only a single buyer, will approximate the target benchmark.  In particular, we show how to choose a number of licenses and a price floor so that the worst-case welfare, at any equilibrium, is a constant approximation to the best achievable welfare under truth-telling after excluding the welfare contribution of a single buyer.

\end{abstract}
\newpage


\section{Introduction}


%
%


Licenses for carbon and other emissions are a market-based policy tool for reducing pollution and mitigating the effects of climate change. Roughly speaking, a government agency distributes pollution licenses to firms according to some mechanism.  At the end of a period of time (e.g., a year), firms must submit licenses to cover their pollution or else face severe penalties.  Different ways of distributing licenses are possible.  For example, if licenses are simply sold at a fixed price to anyone who wishes to pay, then this is equivalent to a \emph{carbon tax} where polluters must pay a linear fee to offset their emissions.  Alternatively, in a \emph{cap-and-trade} mechanism, the agency releases a fixed number $C$ of pollution licenses, either via a pre-determined allocation at no price (e.g., based on prior pollution or industry averages) or via auction, and then polluters can trade these licenses on the open market.  Many emission trading systems take features of both schemes, combining a cap-and-trade market with a \emph{price floor} $\pf$, where each license must be sold above the reserve price of $\pf$, and/or a \emph{price ceiling} $\pc$, where an extra license beyond the cap can always be purchased at a price of $\pc$.  In general, such a pollution license market is referred to as an Emission Trading System (ETS).  

There are many important ETS in effect today.  The EU ETS has thousands of participating firms and has raised billions of dollars in auction revenue over the past 10 years~\cite{EUETS}.  The Western Climate Initiative runs a license auction that is linked between California and Quebec~\cite{WCI,CARB}, and the Regional Greenhouse Gas Initiative (RGGI, pronounced `Reggie') serves New England and the New York region~\cite{RGGI}.  These markets differ in the details of their implementations, but in each case,
licenses are distributed on a regular schedule, with a significant quantity of those licenses sold at auction.  

One can model such a license auction as a multi-unit auction; that is, an auction for multiple identical goods.  There is no bound on the supply of goods, but there is a cost of production corresponding to the social cost of more pollution.  The social cost is typically assumed to be increasing and convex, and the value for licenses for each buyer is commonly assumed to be increasing and concave~\cite{weitzman74}.  Each of the ETS described above uses a uniform-price auction rule to resolve this multi-unit auction.  Such auctions proceed roughly as follows.  Participating firms declare bids, which take the form of a non-decreasing concave function that describes their willingness-to-pay for varying quantities of licenses.  This can alternatively be viewed as a list of non-increasing marginal bids for each successive license.  Any marginal bids below the reserve price $\pf$ are removed, and licenses are then distributed to the highest remaining marginal bids while supplies last.  Each firm then pays a fixed price $p$ per license, where $p$ is set to some value between the lowest marginal winning bid and the highest marginal losing bid.  Such auctions are not truthful, but are common in practice due to their many advantageous properties; see, e.g., Chapter 7 of~\cite{milgrom_2004}.

Implementing a uniform-price auction for carbon licenses presents an optimization challenge: good outcomes require that the system designer correctly sets the quantity of licenses to distribute and/or the price (either direct price or auction reserve) at which they will be sold.  The goal is to maximize \emph{social welfare}: the aggregate value that firms receive for their licenses (i.e., by producing goods) minus the externality on society caused by polluting the corresponding amount during production.  This optimization problem is complicated by uncertainty.  Even if the social cost of pollution is fully known, the designer may not know what the demand for licenses will be, which makes it hard to predict the optimal level of pollution to allow.  The goal of the designer, then, is to set the parameters of the uniform price auction to maximize the expected outcome over uncertainty in the market.  Notably, the presence of social costs significantly increases the complexity of multi-unit auctions, since efficiency becomes a mixed-sign objective.  Even a slight misallocation of the licenses, resulting in a small reduction of received value, might have a disproportionately large effect on net welfare.  In practice, inefficient allocations may be partially resolved in the trade phase of cap-and-trade systems.  However, there are significant trading frictions, so it is imperative to choose an initial allocation, via the auction, that is as efficient as possible. 

Since uniform price auctions are not truthful, one must account for strategic behavior of participating firms.  A notable feature of ETS markets in practice is that, despite their size, they commonly have a small number of participants with significant market power.\footnote{For example, in the EU ETS, the top 10 firms together control approximately 30 percent of all licenses allocated and traded; see~\cite{Cotton15}, page 127.}
The presence of dominant players suggests imperfect competition and raises the issue of strategic manipulation, in which large individual firms try to influence the market in their favor.  One particular concern is demand reduction, where firms reduce their bids to suppress the price determined by the auction.  We ask: how should one set auction parameters to (approximately) maximize welfare, in the face of strategic bidding?

\subsection{Our Results}



In this paper we apply ideas from theoretical computer science to approach the design of carbon license auctions with strategic firms.  In this new domain, we use facts regarding uniform price auctions to understand what is happening in these markets and probabilistic analysis to handle the uncertainty of valuation realizations.  We then apply a Price of Anarchy analysis to give strategic guarantees.  As a result, we provide a concrete recommendation for how to set the parameters of the mechanism used in practice---a recommendation with provable approximation guarantees.

%
%
A market instance is described by a convex social cost curve and, for each participating firm, an independent distribution over concave valuations.  To capture the relevant space of auctions, we formalize a class of allocation rules that we call \emph{cap-and-price auctions}.\footnote{Note, we are focused on the auction phase in this paper and do not model the trade phase of these systems.}  These auctions are parameterized by a maximum number of licenses to allocate, a price floor, and a price ceiling.  They proceed by running a uniform-price auction for the given number of licenses, subject to the price floor and ceiling.  Given a market instance, there is some choice of cap-and-price auction that maximizes the expected welfare generated under non-strategic (truth-telling) behavior.  This optimal non-strategic solution will serve as a benchmark, which we denote $\opt$.  Our main result is an auction that approximates $\opt$ at any equilibrium of strategic behavior.

We show that, in general, the welfare obtained at an equilibrium of a cap-and-price auction can be negative, even for the auction that optimizes welfare under non-strategic bidding.  This motivates our main question: given a cap-and-price mechanism that achieves a certain expected welfare $W$ under truth-telling behavior, can one modify the auction parameters so that the expected welfare is an $O(1)$ approximation to $W$ at any Bayes-Nash equilibrium?   To answer this question, we describe a subclass of cap-and-price auctions that use \emph{safe prices}.  A \emph{safe-price auction} is a cap-and-price auction in which the price floor is at least the average social cost, per license, of selling all of that auction's licenses.  We show that for any such auction the worst-case welfare at any Bayes-Nash equilibrium is within a constant factor of the welfare under truth-telling.  This result proceeds by transforming the market instance into a corresponding instance without social costs, and makes use of bounds on the Bayes-Nash price of anarchy for uniform-price auctions~\cite{MarkakisT15,deKeijzerMST13}.

To construct an auction that approximates $\opt$ at equilibrium, it therefore suffices to find a safe-price auction whose non-strategic welfare approximates that of the best unconstrained cap-and-price auction.  However, as it turns out, there are market instances for which no constant approximation by a safe-price auction is possible.
Such a situation can occur if there is a single firm that drives most of the demand for licenses, and this demand has very high variance.  We show that this is essentially the only barrier to our desired price of anarchy result: one can always construct an auction that achieves a constant approximation to $\opt$ minus the welfare obtainable by allocating licenses to any one single auction participant.  In other words, as long as no single firm accounts for most of the expected welfare of the carbon license market, one can find a cap and price floor so that the welfare at any equilibrium is within a constant of the optimal outcome achievable without strategic behavior.  We argue that this assumption is reasonable, since presumably the motivation for running a carbon license market in the first place is that the demand for emission licenses is distributed across multiple firms.  An alternative way to view this result is that one can always achieve a constant approximation to our benchmark by either running a safe-price auction, or by selling to only a single firm.  

Along the way, we also show that for any cap-and-price auction that sells all of its licenses with probability $q > 0$, there is a safe-price auction whose welfare (under truth-telling) is within a factor $O(1/q)$ of the original auction's welfare.  In particular, if our benchmark is implemented by a cap-and-price auction that distributes all of its licenses with constant probability, then we can construct a safe-price auction whose welfare at any Bayes-Nash equilibrium is a constant approximation to the benchmark.

Note that our results are applicable for any setting with imperfect competition where a uniform price auction is used to allocate goods with a production or social cost.  For instance, they may apply to allocating queries to a database, where the social cost is privacy.  Another example might be allocating medallions to taxi and rideshare drivers in New York City, where the social cost is congestion.

\subsection{Related Work}

Our work is related to rich line of literature in economics comparing emission licence auction formats and flat ``carbon tax'' pricing methods to control the externalities of pollution.  Weitzman~\cite{weitzman74} proposed a model of demand uncertainty and initiated a study comparing price-based vs quantity-based screening in the context of pollution externalities.  Kerr and Cramton~\cite{KerrCramton98} noted that auctions tend to generate more efficient outcomes than grandfathered contracts (i.e., pre-determined allocations based on prior usage), which distort incentives to reduce pollution and efficiently redistribute licenses.  Cramton, McKay, Ockenfels and Stoft~\cite{CramtonBook17} subsequently lay out arguments in favor of tax-based approaches.  Murray, Newell, and Pizer~\cite{Murray09} analyze the use of price ceilings in emission license auctions, and argue that they provide benefits of both auction-based and tax-based systems, improving efficiency in dynamic markets with intertemporal arbitrage.   For a recent overview of auction-based systems used in practice, from both the economic and legal perspectives, we recommend~\cite{Cotton15}.

A similar line of literature studied alternative approaches to the related electricity markets, where individual providers sell electricity into a central grid.  Whereas the main issue in emission license sales is the social externality of production, the main focus in electricity markets was incentivizing participation of small firms.  The primary discussion focused on using uniform pricing versus discriminatory pricing in the resulting procurement auction~\cite{CramtonS07,RassentiSW03}.

Within the theoretical computer science community, there have been numerous studies of the efficiency of auction formats at equilibrium.  For multi-unit auctions without production costs (or, equivalently, no social cost of pollution), the price of anarchy of the uniform-price auction was shown to be constant for full-information settings, and this was subsequently extended to Bayes-Nash equilibria~\cite{MarkakisT15,deKeijzerMST13}.  Our work can be seen as an extension of that work to the setting with a convex cost of allocation.
Auctions with production costs have also been studied~\cite{BlumGMS11,HuangK15}, but primarily from the perspective of mechanism design, where the goal is to develop allocation and payment rules to achieve efficient outcomes.  In the present work we do not take a mechanism design approach; we instead restrict our attention to (non-truthful) uniform price auctions, as these are used in practice, and study bounds on the worst-case welfare at equilibrium under different choices of the auction parameters.

Uniform-price auctions with costs can be viewed as games for which the designer has a mixed-sign objective (i.e., total value generated minus the social cost of production).  Prior work on the price of anarchy under mixed-sign objectives has focused primarily on routing games~\cite{Vetta02,ChauS03,ColeDR12}.  Results in this space tend to be negative, motivating alternative measures of performance (such as minimizing a transformed measure of total cost) that avoid the pitfalls of mixed-sign optimization.  Our work shows that in an auction setting, use of an appropriately-chosen thresholding rule (in the form of a reserve price) can enable a constant approximation to our mixed-sign objective of total value minus a convex social cost.

There is a vast amount of work on auctions with externalities, such as the seminal work of \cite{jehiel1996not}, and other work including externalities in advertising auctions~\cite{Ghosh:2008:EOA:1367497.1367520}, characterizations of equilibria in auctions with externalities~\cite{leme2012sequential}, and more.  However, these externalities are private and held by the buyers, as opposed to public and suffered by the seller as in this paper.  In this case, the externality functions more like a production cost, as described above.

Kesselheim, Kleinberg and Tardos~\cite{KesselheimKT15} study the price of anarchy of an energy market auction, which is a similar application to the carbon license auction that we study.  Their focus is on the uncertainty of the supply and temporal nature of the auction, and not on externalities of production, and hence their technical model is quite different.

The theory of bidding in uniform-price auctions with imperfect competition is well-developed in the economic theory literature.  Much of this work focuses specifically on the efficiency and revenue impact of demand reduction, and how modifications to the auction format or context might impact it.  Demand reduction occurs at equilibrium even in very simple settings of full information, can dramatically reduce welfare, and is a concern in practice~\cite{Wilson79,Weber97,AusubelCPRW14}.  One way to reduce the inefficiency of demand reduction is to perturb the supply, either by allowing the seller to adjust the supply after bids are received~\cite{McAdams07}, or by randomizing the total quantity of goods for sale (or otherwise smoothing out the allocation function)~\cite{KremerN04}.  Such results are typically restricted to full-information settings.  Moreover, these approaches are not necessarily appropriate in the sale of government-issued licenses, where one typically expects commitment and certainty about the quantity being sold.  In contrast, we forego a precise equilibrium analysis and instead argue that setting a sufficiently high price floor can likewise mitigate the impact of demand reduction.

\subsection{Roadmap}

Our main result is a cap-and-price auction whose welfare in equilibrium approximates the maximum welfare of a cap-and-price auction with non-strategic reporting in markets without a single dominant bidder.  In Section~\ref{sec:prelims}, we introduce our model and formally define cap-and-price auctions.  In Section~\ref{sec:uniformprice}, we first show that we can restrict attention to cap-and-price auctions with an infinite price ceiling (i.e., ones who never sell more licenses than the cap).  We then derive a class of such cap-and-price auctions, which we call safe-price auctions, whose performance in equilibrium approximates the performance in a non-strategic setting.  Thus it suffices to show that the best safe-price auction approximates the best cap-and-price auction in a non-strategic setting.  However, this is not true: we give an example where it can be unboundedly worse.  In Section~\ref{sec:main}, we demonstrate that the only barrier to this is the existence of a dominant bidder, yielding our main result.

\section{Preliminaries} \label{sec:prelims}

There are $n$ firms seeking to purchase carbon licenses.  Licenses are identical, and each permits one pollution unit.  Each firm $i$ has a monotone non-decreasing concave \emph{valuation curve} $V_i(\cdot)$ that maps a number of pollution units $x \in \Z_{\geq 0}$ to a value $\R_{\geq 0}$; this is the value they enjoy for polluting this amount.
The profile of valuation curves $\mathbf{V} = (V_1, \dotsc, V_n)$ is drawn from a publicly-known distribution $F$ over profiles, where the draw of $V_i(\cdot)$ is the private information of firm $i$.  Valuations are drawn independently across firms, so $F = F_1 \times \dotsc \times F_n$ is a product distribution and 
$V_i \sim F_i$.

There is a publicly known monotone non-decreasing convex \emph{cost function} $Q(\cdot)$ that maps a number of pollution units $x \in \Z_{\geq 0}$ to the value of externality that society faces for having this much pollution.
Given a valuation profile $(V_1, \dotsc, V_n)$, the \emph{welfare} of a given allocation rule $\textbf{x} = (x_1, \dotsc, x_n)$ is $\sum_i V_i(x_i) - Q(\sum_{i=1}^n x_i)$.  Our objective is to maximize expected welfare.

Given some integer $x \geq 0$, we write $V(x)$ for the maximum aggregate value that could be obtained by optimally dividing $x$ licenses among the firms.  That is, 
$$V(x) = \max_{\vec{y} \in \Z_+^n : ||y||_1 = x} \sum_i V_i(y_i).$$
We refer to $V$ as the \emph{combined valuation curve}.
As each $V_i(\cdot)$ is weakly concave, so is $V(\cdot)$.  
We will sometimes abuse notation and write $V \sim F$ to mean that $V$ is distributed as the aggregate value when $(V_1, \dotsc, V_n) \sim F$.  We'll write $W(V,x) = V(x) - Q(x)$ for the welfare generated by allocating $x$ licenses optimally among the firms.

We'll write $v_{i}(j) = V_i(j) - V_i(j-1)$ for firm $i$'s marginal value for aquiring license $j$, for each $j \geq 1$.  By the monotonicity and concavity of $V_i(\cdot)$, $v_{i}(\cdot)$ is non-increasing.  We'll also write $V_i(j|k) = V_i(j+k) - V_i(k)$ for the marginal value of $j$ additional items given $k$ already allocated.  
We'll write $d_i(p) = \max\{ j \colon v_i(j) \geq p \}$ for the number of units demanded by bidder $i$ at price $p$.  



We will study equilibria and outcomes of license allocation auctions.  An auction takes as input a reported valuation $V_i$ from each firm $i$.  The auction then determines an allocation $\textbf{x} = (x_1, \dotsc, x_n)$ and a price $p_i \geq 0$ that each firm must pay.  The auctions we consider will be uniform price auctions, where the auction determines a per-license price $p$ and each firm $i$ pays $p_i = p \cdot x_i$.  Given an implicit uniform-price auction, we will tend to write $x_i(\mathbf{\tilde{V}})$ (resp., $p(\mathbf{\tilde{V}})$) for the allocation to agent $i$ (resp., per-unit price) when agents report according to $\mathbf{\tilde{V}}$.  For a given valuation profile $\mathbf{V}$, we'll also write $U_i( \mathbf{\tilde{V}} )$ for the utility enjoyed by firm $i$ when agents bid according to $\mathbf{\tilde{V}}$:
\[ U_i(\mathbf{\tilde{V}}) = V_i(x_i(\mathbf{\tilde{V}})) - p(\mathbf{\tilde{V}}) \cdot x_i(\mathbf{\tilde{V}}). \]
Finally, given an auction $M$ and a distribution $F$ over input profiles, we will write $W(M, F)$ for the expected welfare of $M$ on input distribution $F$.  We will sometimes drop the dependence of $F$ and simply write $W(M)$ when $F$ is clear from context.  We emphasize that $W(M)$ is a non-strategic notion of welfare; it is the expected welfare of $M$ with respect to \emph{inputs} drawn from $F$, or equivalently the welfare of $M$ under truthful reporting when agent valuations are distributed according to $F$.  Let $\textbf{x}^M(V) = (x_1^M(V), \ldots, x_n^M(V))$ be the allocation rule of mechanism $M$ for realization $V$.  Then
\[W(M) = \E_{V \sim F}\left[\sum_i V_i\left(x_i^M(V)\right) - Q\left(\sum_i x_i^M(V)\right)\right].\]

A \emph{Bayes-Nash equilibrium} of a uniform-price auction is a choice of bidding strategies $\sigma = (\sigma_1, \dotsc, \sigma_n)$ for each agent, mapping each realized valuation curve $V_i$ to a (possibly randomized) reported valuation $\sigma_i(V_i)$, so that each agent maximizes their expected utility by bidding according to $\sigma_i$ given that other agents are bidding according to $\sigma_{-i}$.  That is, for all $V_i$ and all $\tilde{V}_i$, we have
\[ \E_{V_{-i} \sim F_{-i}}[ U_i(\sigma_i(V_i), \sigma_{-i}(V_{-i})) ] \geq \E_{V_{-i} \sim F_{-i}}[ U_i(\tilde{V}_i, \sigma_{-i}(V_{-i})) ].  \]

We make a standard assumption of \emph{no overbidding} on behalf of the firms uniform-price auction, which is motivated by the fact that overbidding is a weakly-dominated strategy.




Motivated by license auctions used in practice, we study a particular form of uniform price auction that we call a \emph{cap-and-price auction} which forces prices to be within some fixed interval.\footnote{Note the cap is not a hard constraint, but rather governs which prices bind.} 

\begin{definition}\label{def:capandprice}
A {\em cap-and-price auction} $M(C,\pf,\pc)$ is parameterized by a quantity $C \geq 1$, a price floor $\pf$, and a price ceiling $\pc > \pf$.  The allocation and price are determined as follows:
\begin{enumerate}
	\item If $\sum_i d_i(\pc) \geq C$, then $x_i = d_i(\pc)$ for all $i$ and $p = \pc$. \label{cp:pc}
	\item If $\sum_i d_i(\pf) < C$, then $x_i = d_i(\pf)$ for all $i$ and $p = \pf$. \label{cp:pf}
	\item Otherwise, choose $p = V(C) - V(C-1)$, the $C^{\mathrm{th}}$ highest bid, and choose $\vec{x}$ to be any optimal allocation of $C$ licenses among the firms: $x \in \arg\max_{\vec{y} \in \Z_+^n : ||y||_1 = C} \sum_i V_i(y_i)$. \label{cp:cap}
\end{enumerate}	
\end{definition}

We can think of this as a uniform price auction of up to $C$ licenses (case \ref{cp:cap} above), where the price is set to the lowest winning marginal bid, with two modifications.  First, there is a reserve price $\pf$, so that possibly fewer than $C$ licenses are sold if there is not enough demand at price $\pf$; this is case \ref{cp:pf}.  Second, if the lowest winning marginal bid would be larger than $\pc$, then the price is lowered to $\pc$ and firms can purchase as many licenses as they like at this price; this is case \ref{cp:pc}.


Given social cost function $Q$ and valuation distribution $F$, we will write $\opt$ for the optimal expected welfare obtained by any cap-and-price auction under truthful reporting.  That is, $\opt = \max_{C,\pf,\pc}\{ W(M(C,\pf,\pc)) \}$.  We will also tend to write $C^*$, $\pf^*$, and $\pc^*$ for the parameters that achieve this maximum.  We emphasize that this is a non-strategic notion: $\opt$ is the maximum expected welfare attainable when bidders report truthfully.  We view $\opt$ as a benchmark against which we will compare performance at Bayes-Nash equilibrium.  Note also that, in general, $\opt$ may be strictly less than the expected welfare of the unconstrained welfare-optimal allocation; an example is provided in Appendix~\ref{app:cap.price.bad}.

\section{Safe-Price Auctions} 
\label{sec:uniformprice}

We will derive a class of cap-and-price auctions, which we will call safe-price auctions, whose performance in equilibrium approximates their non-strategic welfare.  The hope is that the non-strategic welfare of this class then approximates the non-strategic welfare of the larger class of cap-and-price auctions. \begin{definition}
	A {\em safe-price auction} is a cap-and-price auction $M(C,\pf,\pc)$ parameterized by a license quantity $C$, a price floor 
	$\pf = Q(C)/C$, and price ceiling $\pc = \infty$.  We will write $M(C)$ for the safe-price auction with license quantity $C$, and $P(C) = Q(C)/C$ for the associated safe price.
\end{definition}

This definition restricts cap-and-price auctions in two ways. The first is that the price ceiling is now infinite.  The second is that we have imposed a lower bound on the price floor.\footnote{In fact we set $\pf$ to be \emph{equal} to $P(C)$ , but our results still hold if this is relaxed to $\pf \geq P(C)$.}  The first restriction is for convenience; the following lemma shows that it does not cause much loss in welfare. Motivated by this lemma, we will restrict our attention from now on to auctions with no price ceiling, that is, $\pc = \infty$.  We notate such mechanisms $M = (C, \pf)$.  Note that, in the language of definition~\ref{def:capandprice}, case (\ref{cp:pc}) can no longer occur.  That is, when is there is no price ceiling, only exactly $C$ or less than $C$ licenses will be allocated. 


\begin{lemma} 
\label{lem.priceceil}
For any distribution $F$ and any cap-and-price auction $M(C,\pf,\pc)$, there exist a cap $C'$ and price floor $\pf'$ such that $W(M(C',\pf',\infty)) \geq \frac{1}{2}W(M(C,\pf,\pc))$.
\end{lemma}

The proof of Lemma~\ref{lem.priceceil} appears in Appendix~\ref{apppf:lem.priceceil}.  The main idea is to decompose the expected welfare of a cap-and-price auction $M(C,\pf,\pc)$ into the welfare attained from the first (at most) $C$ licenses sold plus the incremental welfare attained from any licenses sold after the first $C$.  The first can be approximated by $M(C,\pf,\infty)$; the latter, if non-negative, can be approximated by $M(\infty,\pc,\infty)$. 

The second way safe-price auctions restrict price-and-cap auctions is with a lower bound on the price floor.  As the following example shows, the equilibrium welfare of cap-and-price auctions, even those with an infinite price ceiling, suffer from a problem known as demand reduction in which a bidder improves her price, and hence utility, by asking for fewer units. This can have drastic consequences on welfare, causing it to become negative, due to the existence of the social cost. In particular, this means the approximation of such auctions, relative to the best welfare attainable in non-strategic settings, is unbounded.   

\begin{example}
	\label{ex:poaunbounded}
	Consider two agents, $1$ and $2$.  Their distributions over valuations will be point-masses, so that the valuations are actually deterministic.  These valuation functions are given by the following marginals: $v_1(1) = 10$, $v_1(2) = 10$, $v_2(1) = 6$, $v_2(2) = 1$, and all other marginals are $0$.  We will have $C = 2$, $\pf = 0$, $\pc=\infty$, and the social cost function is given by $Q(x) = 9x$.
	Under truthful reporting, the auction $M(C,\pf)$ allocates $2$ licenses to agent $1$ at a price of $6$ each, resulting in a welfare of $V_1(2) - Q(2) = 20 - 18 = 2$.  However, we note that firm $1$ could improve their utility from $8$ to $9$ by instead reporting a modified valuation $\tilde{V}$ given by $\tilde{v}_1(1) = 10$ and $\tilde{v}_1(2) = 1$.  If agent $2$ continues to report truthfully, $M(C,\pf)$ will allocate $1$ license to each agent at a price of $1$ each, resulting in a welfare of $V_1(1) + V_2(1) - Q(2) = 16 - 18 = -2$.  One can verify that this is indeed a pure Nash equilibrium of the auction, and hence a Bayes-Nash equilibrium as well.
\end{example}

The issue illustrated in Example~\ref{ex:poaunbounded} is that strategic behavior can cause licenses to be allocated to agents whose marginal values for these licenses are below $Q(C) / C$, causing the aggregate value derived from allocating $C$ licenses to fall below $Q(C)$.  Safe-price auctions seek to prevent this by imposing a sufficiently high price floor.
%

Indeed, the following lemma shows that price floors do indeed circumvent the issue in the above example.  Namely, we show safe-price auctions have good equilibria, compared to their own non-strategic welfare.  This is the main result of this section, and motivates us to focus on analyzing the non-strategic welfare of safe-price auctions.


\begin{lemma} 
	\label{lem:safeprice}
	For any safe-price auction $M(C)$, the expected welfare at any Bayes-Nash equilibrium is at least $(\tfrac{1}{3.15}) \cdot W(M(C))$.  
\end{lemma}

\begin{proof} 
	Fix a license cap $C$ and write $\pf = Q(C)/C$ for the corresponding safe price, and recall that $M(C) = M(C,\pf)$.  We claim that this auction is strategically equivalent to a modified auction, as follows.  First, since each agent pays at least $\pf$ per license, we can define $\hat{v}_i(j) = v_i(j) - \pf$ for the residual marginal utility of agent $i$ for their $j^{\mathrm{th}}$ license, after taking into account that they must pay at least $\pf$.  Then from each agent's perspective, playing in $M(C,\pf)$ with marginal values $v_i(j)$ gives the same utility as playing in $M(C,0)$ with marginal values $\hat{v}_i(j)$.  The welfare generated by the original auction is $V(x) - Q(x)$, which is equal to $\hat{V}(x) - ( Q(x) - \pf x)$.  So the welfare generated by $M$ is equivalent to the welfare generated with modified valuations $\hat{V}$, where the marginal social cost of each unit of pollution is reduced by $\pf$.  But now note that from the definition of $\pf$, $Q(x) - \pf x \leq 0$ for all $x \leq C$.  So in an auction with a cap of $C$, the welfare is at least the welfare obtained with a social cost of $0$; it is always non-negative.
	
	We conclude that the welfare generated by $M(C, \pf)$, at any equilibrium, is equal to the welfare generated at an equilibrium of the standard uniform price auction with modified valuations $\hat{V}$.  Theorem 3 of \cite{deKeijzerMST13} (which bounds the welfare in Bayes-Nash equilibria of uniform auctions without costs) therefore applies\footnote{The proof of the 3.15 bound in \cite{deKeijzerMST13} is actually for a uniform-second-price auction that charges the highest losing bid, or $V(C+1) - V(C)$; however, the proof is also correct for charging a price of the lowest winning bid $V(C) - V(C-1)$.}, and the welfare at any equilibrium is at least $\tfrac{1}{3.15}$ times the optimum achievable under the modified valuations, which is at most $W(M(C))$.
\end{proof}



\section{Welfare Approximation}
\label{sec:main}

Motivated by Lemma~\ref{lem:safeprice}, we would like to show that the non-strategic welfare of safe-price auctions approximates that of cap-and-price auctions.  Unfortunately, there are cases where this does not hold, as demonstrated by Example~\ref{ex:logscale2}. 
The example consists of just a single firm that dominates the market and whose demand for licenses has high variance.  The firm's value is always large enough that even at a low price and high license cap, the resulting allocation has high net welfare.  However, safe-price auctions will generate much lower welfare.
Roughly speaking, for any number of licenses $C$ that the auction designer selects as the cap for a safe-price auction, significant welfare will be lost from realizations where the firm's demand is much higher than $C$, and the corresponding safe price will exclude welfare gains from realizations where the firm's demand is much lower than $C$.  If the variance is high enough, this will cause overall welfare to be low regardless of the choice of $C$.



\begin{example} \label{ex:logscale2}
	We will present an example for which the expected welfare of any safe-price auction $M(C)$ is at most an $O(1/n)$-approximation to $\opt$.
	Let $Q(x) = x^2$.  There is a single firm participating in the auction.  That firm's valuation curve is drawn according to a distribution $F$ over $n$ different valuation curves, which we'll denote $V^{(1)}, \ldots, V^{(n)}$.  For all $i$, valuation $V^{(i)}$ is defined by $V^{(i)}(x) = \max\{ 2^{i+1} \cdot x, 2^{2i+1} \}$.  The probability that $V^{(i)}(\cdot)$ is drawn from $F$  is proportional to $\frac{1}{2^{2i}}$.  That is, the firm has valuation $V^{(i)}$ with probability $\frac{1}{\beta 2^{2i}}$, where $\beta = \sum _{i=1} ^n 2^{-2i} = \frac{1}{3} (1-2^{-2n})$ is the normalization constant.
	Note that a cap-and-price auction with cap $C = \infty$ and price floor $\pf = 1$ achieves welfare
	$\sum _{i=1} ^n 2^{2i} \frac{1}{\beta 2^{2i}} = \frac{1}{\beta} \cdot n,$
	so $\opt$ is at least this large.\footnote{In fact, this is the ``first-best'' welfare obtainable at every possible realization, so this is actually the best possible cap-and-price auction and hence the exact value of $\opt$.}
	
	
	Consider a safe-price auction with cap $C'$.  Suppose that $C' \in (2^k, 2^{k+1}]$ for some $k \geq 1$.  Then the safe price is $P(C') = Q(C') / C' = (C')^2/C' = C' > 2^k$.  If the firm has valuation $V^{(i)}$ with $i \leq k-1$, then all marginal values are strictly below $P(C')$, then no licenses are allocated and the welfare generated is $0$.  On the other hand, if the firm has valuation $V^{(i)}$ for any $i \geq k$, then since at most $C' \leq 2^{k+1}$ licenses can be purchased, the auction generates welfare at most $2^{i+1} (2^{k+1}) - (2^{k+1})^2 = 2^{2k+2}(2^{i-k} - 1)$.

	The total expected welfare of $M(C')$ is therefore at most
	\begin{align*}
	\sum _{i=k}^n 2^{2k+2} (2^{i-k}-1) \cdot \frac{1}{\beta 2^{2i}} &= \frac{4}{\beta} \sum _{i=k} ^n  2^{-2(i-k)} (2^{i-k}-1) 
	\leq \frac{4}{\beta} = O(1).
	\end{align*}
	Since $\opt \geq \frac{1}{\beta} \cdot n$ but $W(M(C')) = O(1)$ for all $C'$, we conclude that no auction with a safe price achieves an $o(n)$-approximation to $\opt$.  This concludes the example.
\end{example}



Example~\ref{ex:logscale2} is driven by a single firm that dominates the market and has high variance in their demand.  We next show that this is the \emph{only} barrier to a good approximation: either a safe-price auction is constant-approximate, or else one can approximate the optimal welfare by selling to just a single firm.  To this end, we will be interested in the expected maximum welfare attainable by allocating licenses to just one firm, which we will denote $W^{(1)}$.  That is,
\[ W^{(1)} = \E_{V \sim F}\left[ \max_{i, x \geq 0} \{ V_i(x) - Q(x) \} \right]. \]
The following theorem is our main result.


\begin{theorem}
\label{thm:main}
There exists a constant $c$ such that, for any cap $C$ and price floor $\pf$, there exists $C'$ such that $c \cdot W(M(C')) + W^{(1)} \geq W(M(C,\pf))$.
\end{theorem}

Note that a corollary of Theorem~\ref{thm:main} (combined with Lemmas~\ref{lem.priceceil} and~\ref{lem:safeprice}) is that for any cap-and-price auction $M(C,\pf,\pc)$, there exists some $C'$ such that the expected welfare of any Bayes-Nash equilibrium of $M(C')$ is at least a constant factor of $W(M(C,\pf,\pc)) - W^{(1)}$.  That is, the worst-case equilibrium welfare is a constant-factor approximation to the best welfare achievable by any cap-and-price auction under truth-telling, excluding the welfare contribution of a single firm.

Also, $W^{(1)}$ is the expected welfare of a truthful mechanism that we will call $M^{(1)}$.  In $M^{(1)}$, the firms first participate in a second-price auction for the right to buy licenses.  The firm with the highest bid wins, and they pay the second-highest bid.  The winning firm can then purchase any number of licenses $x \geq 0$, for which they pay $Q(x)$ (in addition to their payment from the initial second-price auction).  Since the utility obtained by firm $i$ if they win the initial auction is precisely $\max_x \{ V_i(x) - Q(x) \}$, and since a second-price auction is truthful and maximizes welfare, we can conclude that $W(M^{(1)}) = W^{(1)}$.  Then another corollary of Theorem~\ref{thm:main} is that one can approximate the non-strategic welfare of any cap-and-price auction using either a safe-price auction (in which case the approximation holds at any Bayes-Nash equliibrium), or using the (truthful) mechanism $M^{(1)}$ that allocates licenses to at most one firm.

\begin{figure}
	\begin{minipage}[t]{0.45\textwidth}
		\centering
		\includegraphics[scale=.262]{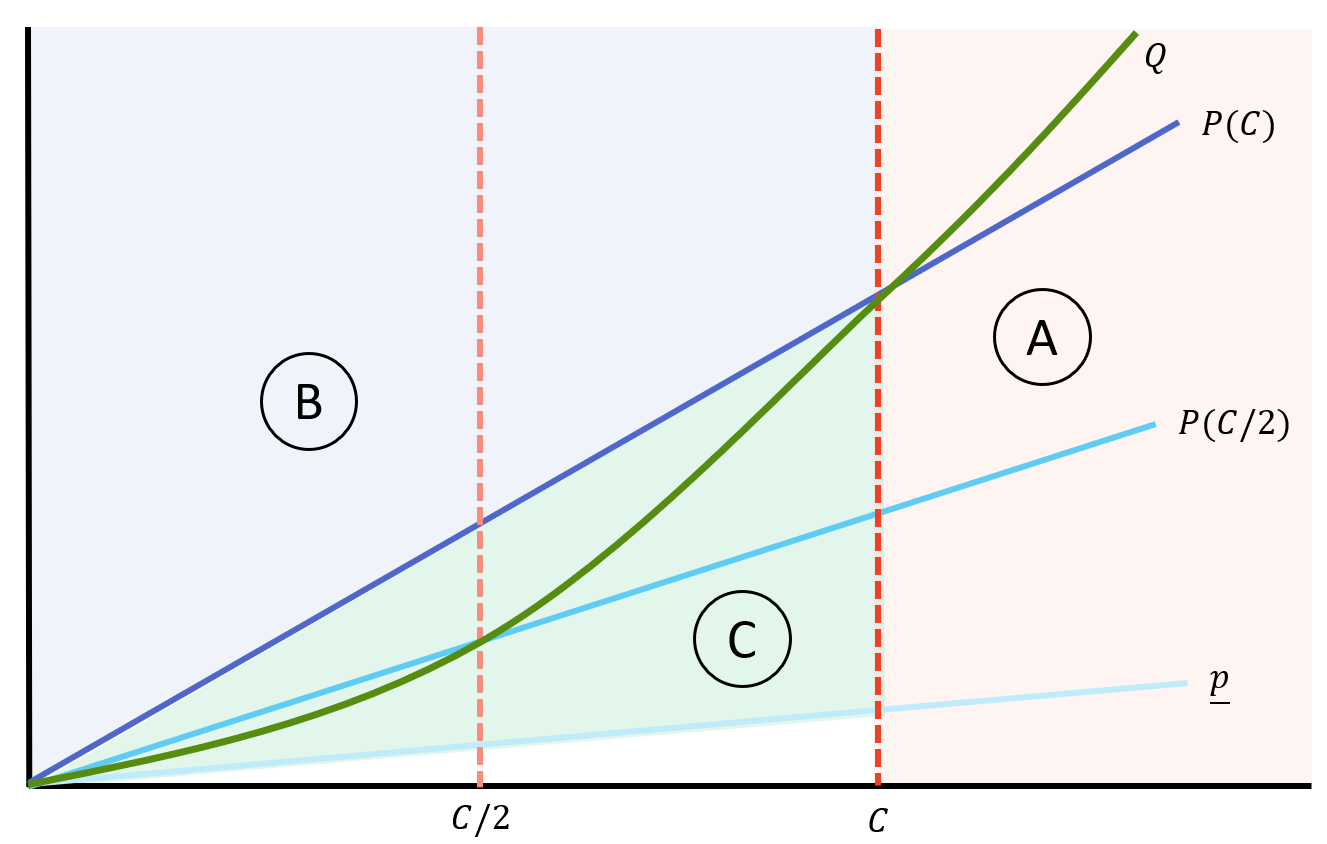}\\(a)
	\end{minipage}
	\hfill
	\begin{minipage}[t]{0.45\textwidth}
		\centering
		\includegraphics[scale=.262]{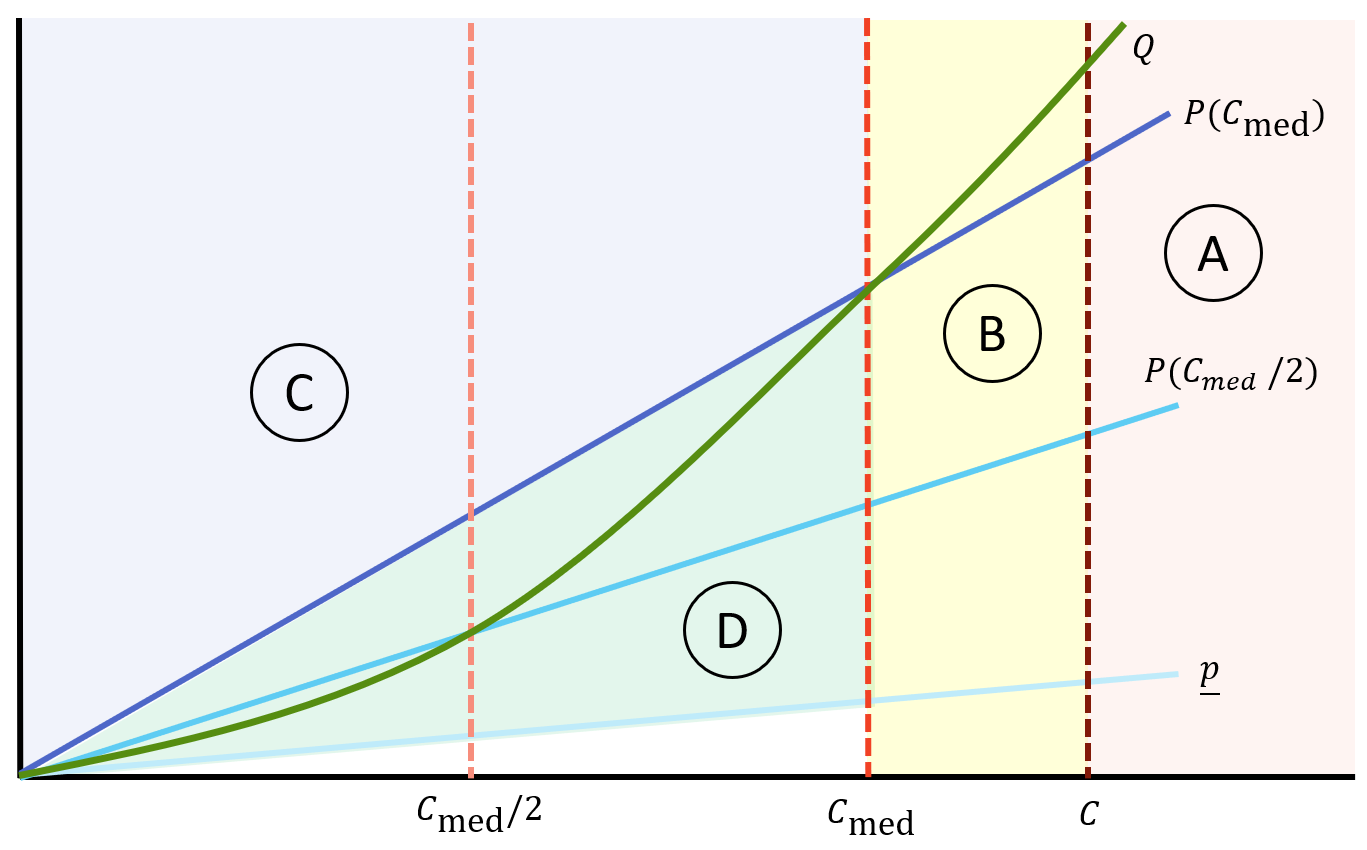}\\(b)
	\end{minipage}
	\caption{A visualization of the proofs of Theorem~\ref{thm.q} (a) and Theorem~\ref{thm:main} (b).  The $x$-axis represents quantity of licenses; the $y$-axis represents value.  The expected welfare of benchmark $W(M(C,\pf))$ is divided into separate contributions, based on the location of allocation outcome $(x,V(x))$; these are depicted as shaded regions.  Each of these contributions is bounded by either the welfare of a safe-price auction or the welfare obtained from a single buyer.  Figure (a) is the (simpler) partition used in Theorem~\ref{thm.q}, and (b) is the partition used in Theorem~\ref{thm:main}.}  \label{fig:independent}
\end{figure}


We are now ready to prove Theorem~\ref{thm:main}.  Fix a cost function $Q$ and distribution $F$, and choose the optimal cap $C$ and price floor $p$ such that $W(M(C,p))$ is maximized.  Let $d(\cdot)$ and $x(\cdot)$ be the demand and allocation functions for $M(C,p)$.  For any given realization of preferences $V$, there is a total demand $d(V) = \sum_i d_i(p)$, where again, $d_i(p) = \argmax_x V_i(x) - px$.  Recall also that $W(x,V) = V(x) - Q(x)$ is the welfare generated by allocation $x$ given aggregate valuation $V$.  

We will begin by proving a simpler version of Theorem~\ref{thm:main}, which is parameterized by $q := \Pr[ d(V) \geq C ]$, the probability that the auction sells $C$ licenses.  When $q$ is bounded away from $0$, we show that there is a safe-price auction that achieves an $O(1/q)$-approximation to $W(M(C,p))$.  We note that this result does not require firm valuations to be independent, and holds even if the valuations are drawn from an arbitrarily correlated distribution $F$.  We will explain in Section~\ref{sec:general} how to extend the argument to the general case of Theorem~\ref{thm:main}, with details deferred to the appendix.

\begin{theorem}
\label{thm.q}
For the welfare-optimal cap $C$ and price floor $p$, there exists a cap $C'$ such that $(1 + 2/q) \cdot W(M(C')) \geq W(M(C,p))$.
\end{theorem}
\begin{proof}

%

%

The welfare generated by $M(C,p)$ can be broken down as follows: 
\begin{align*}
W(M(C,p)) &= \Pr_{V \sim F}[d(V) \geq C] \cdot \E[W(x(V),V) \ |\ d(V) \geq C] \\
&\quad + \Pr_{V \sim F}[d(V) < C] \cdot \E[W(x(V),V) \ |\ d(V) < C]. 
\end{align*}
The first term is the contribution to the welfare from events where $C$ licenses are sold (or, equivalently, at least $C$ licenses are demanded).  The second term is the contribution from events where the total demand for licenses is less than $C$.

Consider price $P(C)$, the safe price for quantity $C$.  We wish to decompose the second term in the expression above into marginal values above $P(C)$ and those below $P(C)$.  For this, we need to introduce some notation.  Write $\theta_i$ for the largest $j \geq 1$ such that $v_i(j) \geq P(C)$, or $\theta_i = 0$ if $V_i(1) < P(C)$.  
Write $x^{>}_i = \min\{x_i, \theta_i\}$, and write $x^{<}_i = x_i - x^{>}_i$.  
Then $x_i^{>}$ is the part of allocation $x_i$ for which firm $i$ has marginal value at least $P(C)$ per unit, and $x_i^{<}$ is the part of $x_i$ for which firm $i$ has a marginal value less than $P(C)$ per unit.
We'll also define $V_i^{<}$ as $V_i^{<}(x) = V_i^{<}(x | \theta_i)$; this is the marginal value of $i$ for receiving $x$ additional licenses after already having received $\theta_i$ licenses.
Note then that $V_i(x_i) = V_i(x^{>}_i) + V_i^{<}(x^{<}_i)$, for all $i$.  
We then have
\begin{align*}
& \Pr_{V \sim F}[d(V) < C] \cdot \E[W(x(V),V) \ |\ d(V) < C] \\
& \quad = \int_{V \colon d(V) < C} \left(\sum_i V_i(x_i(V)) - Q(x(V))\right) dFV \\
& \quad \leq \int_{V \colon d(V) < C} \left(\sum_i V_i(x_i^{>}(V)) - Q(x^{>}(V))\right) dFV \\
& \quad \quad \quad +  \int_{V \colon d(V) < C} \left(\sum_i V_i^{<}(x_i^{<}(V)) - Q(x^{<}(V)) \right) dFV \\
\end{align*}
where the final inequality used the fact that $Q(x) \geq Q(x^{<}) + Q(x^{>})$, due to the convexity of $Q$.  We conclude that
\begin{align}
W(M(C,p)) &\leq \Pr_{V \sim F}[d(V) \geq C] \cdot \E[W(x(V),V) \ |\ d(V) \geq C] \label{1A}\tag{1A}\\
&\quad +  \int_{V \colon d(V) < C} \left(\sum_i V_i(x_i^{>}(V)) - Q(x^{>}(V)) \right) dFV \label{1B}\tag{1B}\\
&\quad +  \int_{V \colon d(V) < C} \left(\sum_i V_i^{<}(x_i^{<}(V)) - Q(x^{<}(V)) \right) dFV \label{1C}\tag{1C}
\end{align}
so that $W(M(C,p)) \leq \eqref{1A} + \eqref{1B} + \eqref{1C}$.  See Figure~\ref{fig:independent}(a) for an illustration.  

We claim that the welfare obtained from the first two terms, $\eqref{1A}$ and $\eqref{1B}$, are covered by the safe-price auction with cap $C$.  The intuition is that a price floor of $P(C)$ does not interfere with the welfare contribution due to licenses with marginal values greater than $P(C)$.  One subtlety is that some of the licenses in the summation $\eqref{1A}$ might have marginal values less than $P(C)$, but it turns out that it can only improve welfare to exclude such licenses from the allocation.
\begin{claimnum}
\label{claim.1a.1b}
$W(M(C)) \geq \eqref{1A} + \eqref{1B}$.
\end{claimnum}
\begin{proof}
Write $d'$ and $x'$ for the demand and allocation under auction $M(C)$, respectively.  We have
\begin{align*}
W(M(C)) &= \Pr_{V \sim F}[d(V) \geq C] \cdot \E[W(x'(V),V) \ |\ d(V) \geq C] \\
&+ \Pr_{V \sim F}[d(V) < C] \cdot \E[W(x'(V),V) \ |\ d(V) < C]. 
\end{align*}
Note that we intentionally use $d(V)$, the demand for $M(C,p)$, rather than $d'(V)$ in the expressions above.
The second term is precisely \eqref{1B}, since $x'_i = x^{>}_i$ from the definition of $x^{>}$.  We claim that the first term is at least $\eqref{1A}$.  To see why, fix any $V$ with $d(V) \geq C$ (and corresponding allocation $x = x(V)$ where necessarily $|x| = C$), and note that $x'_i \leq x_i$ for all $i$.  Furthermore, $V_i(x'_i) \geq V_i(x_i) - (x_i - x'_i)P(C)$, since $x'_i$ is simply $x_i$ after possibly excluding some items with marginal value less than $P(C)$.  Thus, since $|x| = C$, 
\[ V(x') - |x'| \cdot P(C) \geq V(x) - |x| \cdot P(C) = V(x) - Q(|x|). \]
Also, by convexity, we have $Q(y) \leq y \cdot P(C)$ for all $y \in [0,|x|]$.  In particular, we have $V(x') - Q(|x'|) \geq V(x') - |x'| \cdot P(C) \geq V(x) - Q(|x|)$, and hence
\begin{align*} &\Pr_{V \sim F}[d(V) \geq C] \cdot \E[W(x'(V),V) \ |\ d(V) \geq C] \\
&\geq \Pr_{V \sim F}[d(V) \geq C] \cdot \E[W(x(V),V) \ |\ d(V) \geq C]\end{align*}
as claimed.
\end{proof}

The more difficult part of the proof of Theorem~\ref{thm.q} is to account for term $\eqref{1C}$.  This represents the contribution of licenses whose marginal values lie below $P(C)$, and are therefore excluded when the price floor is set to $\pf = P(C)$.  What we show is that the welfare contribution from all such licenses is approximated by $W(M(C/2))$, the welfare obtained by the safe auction with cap $C/2$.\footnote{For convenience we will assume $C$ is even for the remainder of this section.  When $C$ is odd, the result holds for at least one of the floor or the ceiling of $C/2$.  Details appear in the proof of Theorem~\ref{thm:main} in Appendix~\ref{app:proof.thm.main}.}  We do this in three steps.  First, we show that any optimal cap-and-price auction $M(C,\pf)$ must generate non-negative expected welfare conditional on selling $C$ licenses.  The intuition is that if the expected welfare from selling $C$ licenses is negative, then it must be strictly preferable to reduce the quantity of licenses.  

\begin{lemma}\label{lem.opt}
If cap $C$ and price floor $p$ are chosen to maximize $W(M(C,p))$, then 
\[\E_{V \sim F}[W(M(C,p)) \ |\ d(V) \geq C] \geq 0,\] where $d(V)$ is the total quantity of licenses demanded at price $p$.
\end{lemma}

The proof of Lemma~\ref{lem.opt} appears in Appendix~\ref{apppf:lem.opt}.  Next, for any cap-and-price auction $M$, we provide an upper bound on the total welfare contribution that comes from events where $x < C$ licenses are sold, and the average marginal value of the sold licenses, $V(x) / x$, is less than $P(C)$.  In fact, it will be useful to state the lemma more generally: we provide a more general bound that applies to any $\tilde{C} \leq C$ and any valuation that is at most $P(\tilde{C}) \cdot x$ (not necessarily $V(x)$).  The lemma bounds the total welfare contribution, $P(\tilde{C}) \cdot x - Q(x)$, by the difference in price between (a) $\tilde{C}$ licenses sold at the safe price for $\tilde{C}$, and (b) the same $\tilde{C}$ licenses sold at the safe price for $\tilde{C}/2$.

\begin{lemma}\label{lem.unsafe.points}
Choose some quantity $\tilde{C}$ and a number of licenses $x \leq \tilde{C}$.  
Then $P(\tilde{C}) \cdot x  - Q(x) \leq \tilde{C} \cdot ( P(\tilde{C}) - P(\tilde{C}/2) )$.
\end{lemma}

The proof of Lemma~\ref{lem.unsafe.points} appears in Appendix~\ref{apppf:lem.unsafe.points}.  We are now ready to bound the contribution of $\eqref{1C}$.  The intuition is as follows.  By Lemma~\ref{lem.unsafe.points}, the contribution from $\eqref{1C}$ is at most (a constant times) the gap between the line $y = P(C) \cdot x$ and the curve $Q(x)$ at point $x = C/2$.  But recall that $\Pr[ d(V) \geq C ] = q$, and Lemma~\ref{lem.opt} implies that, on average, the expected marginal value of licenses allocated subject to this event is at least $P(C)$.  Therefore, if we set a license cap of $C/2$, then with probability at least $q$ all $C/2$ licenses will be sold at an average expected marginal value of at least $P(C)$.  The welfare generated in this case covers the ``gap'' at $C/2$, and hence covers the loss due to excluding licenses with marginal values at most $P(C)$.

\begin{claimnum}
\label{claim.1c}
$W(M(C/2)) \geq \tfrac{1}{2}q \cdot \eqref{1C}$.
\end{claimnum}
\begin{proof}
Write $d'$ and $x'$ for the demand and allocation under mechanism $M(C/2)$, respectively.  Choose any $V$ such that $d(V) \geq C$.  For any such $V$, $d'(V) \leq d(V)$, and is formed by removing items with marginal value at most $P(C/2)$.  In particular, since $x'(V) \leq d'(V)$, we have $V(x'(V)) \geq V(C) - (C-x'(V))\cdot P(C/2)$.  
Noting that $Q(|x'(V)|) \leq |x'(V)| \cdot P(C)$ since $x'(V) \leq x(V)$, we have 
\begin{align*}
V(x'(V)) - Q(|x'(V)|) &\geq V(x'(V)) - |x'(V)| \cdot P(C) \\
& \geq V(C) - (C-|x'(V)|)\cdot P(C/2) - |x'(V)| \cdot P(C).
\end{align*}
Taking an expectation over all $V$ such that $d(V) \geq C$, we note that Lemma~\ref{lem.opt} implies $\E[V(C)] \geq C \cdot P(C)$.  We therefore have
\begin{align*}
\E[V(x'(V)) - Q(|x'(V)|)]
&\geq C \cdot P(C) - (C - \E[|x'(V)|]) \cdot P(C/2) - \E[|x'(V)|] \cdot P(C) \\
& = (C - \E[|x'(V)|])(P(C) - P(C/2)) \\
& \geq (C/2) \cdot (P(C) - P(C/2))
\end{align*}
where in the last inequality we used that $|x'(V)| \leq C/2$ by definition.  
Since $\Pr[d(V) \geq C] = q$, we therefore have that
\[ W(M(C/2)) \geq \Pr[ d(V) \geq C] \cdot \E[V(x'(V)) - Q(|x'(V)|) \ |\ d(V) \geq C] \geq q (C/2) \cdot (P(C) - P(C/2). \]
We now claim that $(C/2) \cdot (P(C) - P(C/2) \geq (1C)/2$, completing the proof.  To see why, note that for any $x^* < C$ and any $V$, $V^{<}(x^*) - Q(|x^*|) \leq |x^*| \cdot P(C) - Q(|x^*|)$.  Thus $(1C) \leq \max_{x^*}\{ |x^*| \cdot P(C) - Q(|x^*|) \}$.  But note that for any $x^* > C/2$, we have $Q(x^*) > |x^*| P(C/2)$, and hence $|x^*| \cdot P(C) - Q(|x^*|) \leq |x^*| \cdot (P(C) - P(C/2)) \leq (C) \cdot (P(C) - P(C/2))$.  Also, for any $x^* < C/2$, $Q(|x^*|)$ lies above the line joining $(C/2, Q(C/2))$ and $(C,Q(C))$, and hence $|x^*| \cdot P(C) - Q(|x^*|)$ is at most the distance between $|x^*| \cdot P(C)$ and this line, which is at most $(C) \cdot (P(C) - P(C/2))$.  So in either case we have $\max_{x^*}\{ |x^*| \cdot P(C) - Q(|x^*|) \} \leq C \cdot (P(C) - P(C/2))$ as required.
%
%
\end{proof}


Combining Claim~\ref{claim.1a.1b} and Claim~\ref{claim.1c} we have that $W(M(C)) + (2/q)W(M(C/2)) \geq W(M(C,\pf))$, which completes the proof of Theorem~\ref{thm.q}. 
\end{proof}

As a corollary, Theorem~\ref{thm.q} combined with Lemma~\ref{lem:safeprice} and Lemma~\ref{lem.priceceil} implies that for any cap and price auction $M(C,\pf,\pc)$, there exists a safe-price auction $M(C')$ such that, at any Bayes-Nash equilibrium of $M(C')$, the expected welfare generated is at least $\frac{1}{3.15} \cdot \frac{1}{2} \cdot \frac{1}{1+2/q} \cdot W(M(C,\pf,\pc))$.  In particular, if $q$ is a constant bounded away from $0$, then $M(C')$ obtains a constant fraction of $W(M(C,\pf,\pc))$ at any Bayes-Nash equilibrium.

\subsection{The General Case}

\label{sec:general}
We complete the proof of Theorem~\ref{thm:main} in Appendix~\ref{app:proof.thm.main}.  Here we describe at a high level what steps are needed to complete the argument.  We will focus on the case where $\Pr[ d(V) \geq C ] < 1 - 1/e$, since if $\Pr[ d(V) \geq C ] \geq 1 - 1/e$ then Theorem~\ref{thm:main} follows from Theorem~\ref{thm.q}.

Recall that in Claim~\ref{claim.1c}, we used the assumption that $\Pr[ d(V) \geq C ] = q$ to argue that the welfare gained in $M(C/2)$ in the event that $d(V) \geq C$ covers the welfare lost from marginals lying below $P(C)$, up to a constant factor.  If the probability that $d(V) \geq C$ is small, then this may no longer be true.  To handle this, we consider a reduced cap $\Cmed < C$ set to be the largest integer such that $\Pr[ d(V) \geq \Cmed ] \geq 1 - 1/e$.  Our hope is to reproduce the argument from Claim~\ref{claim.1c}, but substituting $\Cmed$ for $C$.  To this end, we divide the welfare of $M(C,\pf)$ into \emph{four} parts: all welfare under the event that $d(V) > C$; all welfare from individual agents whose demand is at least $\Cmed$; the contribution of marginal values greater than $P(\Cmed)$ (but with individual firms demanding at most $\Cmed$) when $d(V) < C$, and the contribution of marginal values less than $P(\Cmed)$ when $d(V) < C$.  See Figure~\ref{fig:independent}(b) for an illustration.  

As in the proof of Theorem~\ref{thm.q}, the contribution due to events where $d(V) > C$ can be covered by the welfare of $M(C, P(C))$.  

The contribution from agents who individually demand at least $\Cmed$ licenses is a new case that we didn't have to handle in Theorem~\ref{thm.q}.  It is here that we use $M^{(1)}$, allocating to any single agent.  Because the \emph{total} demand is at most than $\Cmed$ with probability at least $1/e$, independence implies that the expected number of agents who demand more than $\Cmed$ licenses is at most $1$, given that the probability that none have demand more $\Cmed$ must be at least $1/e$.  So the total contribution to welfare of all such events is at most $W(M^{(1)})$.

If $d(V) \leq \Cmed$, then the contribution from marginal values that are at least $P(\Cmed)$ can be covered by the welfare of $M(\Cmed, P(\Cmed))$, precisely as in Claim~\ref{claim.1a.1b}. We must also handle the case that $d(V) \in [\Cmed, C]$, and consider the welfare contribution of agents that do not (individually) demand more than $\Cmed$ licenses.  Here we use independence: the total quantity demanded by such ``small'' agents is likely to concentrate, so it is unlikely that the total demand will be larger than $2\Cmed$.  Thus, by imposing a cap of $\Cmed$, we lose at most a constant factor of the welfare from marginal values greater than $P(\Cmed)$.

The final step is to show that there is some $C'$ such that $W(M(C'))$ obtains at least a constant fraction of the welfare generated by marginal values less than $P(\Cmed)$, similarly to Claim~\ref{claim.1c}.  If $\Cmed$ is even then we'll take $C' = \Cmed/2$; otherwise, $C' \in \{ \lfloor \Cmed \rfloor, \lceil \Cmed \rceil\}$.  When proving Claim~\ref{claim.1c}, we used Lemma~\ref{lem.opt} to argue about the welfare generated by events where $d(V) > C$.  Unfortunately, Lemma~\ref{lem.opt} does not extend to $\Cmed$: it could be that the expected welfare generated by $M(\Cmed, P(\Cmed))$, conditional on selling $\Cmed$ licenses, is negative.  However, we \emph{can} prove an upper bound on how negative this expected welfare can be.  After all, if the expected welfare is sufficiently negative sufficiently often, it would be welfare-improving to increase the price floor of $M(C,\pf)$ from $\pf$ to $P(\Cmed)$, contradicting the supposed optimality of $M(C,\pf)$.  This turns out to be enough to prove a bound similar to Claim~\ref{claim.1c}.

Combining these bounds, we can conclude that each of the four parts of the welfare of $M(C, \pf)$ can be covered by either a safe-price auction or by $M^{(1)}$, which completes the proof of the theorem.

\bibliographystyle{plain}
\bibliography{refscarbon}

\newpage
\appendix
\section{Omitted Proofs from Section~\ref{sec:uniformprice}}
\label{sec:proofs3}

\subsection{Proof of Lemma~\ref{lem.priceceil}} \label{apppf:lem.priceceil}

Recall the statement of Lemma~\ref{lem.priceceil}: for any distribution $F$ and any cap-and-price auction $M(C,\pf,\pc)$, there exist a cap $C'$ and price floor $\pf'$ such that $W(M(C',\pf',\infty)) \geq \frac{1}{2}W(M(C,\pf,\pc))$.

\begin{proof}
We show that given any mechanism $M= (C_1, \pf, \pc)$, 
we can construct a mechanism $M'$ with price ceiling $\infty$ 
such that $W(M') \geq \frac{1}{2}W(M)$. 

We can decompose the expected welfare of $M$ into (a) the welfare attained from the first (at most) $C_1$ licenses sold, and (b) the incremental welfare attained from any licenses sold after the first $C_1$.  

Note that the auction $M^1 = (C_1, \pf, \infty)$, which is $M$ but with price ceiling set to $\infty$, achieves welfare precisely equal to the former of these two parts.  This is because $M^1$ always allocates at most $C_1$ licenses, and will allocate them efficiently subject to all marginal values being at least $\pf$.

Next consider the second of these two parts of the welfare of $M$.  If the expected welfare in the second part is negative, then we are already done, so suppose not.  Whenever more than $C_1$ licenses are sold, all licenses are sold at price $\pc$, and therefore have marginal value at least $\pc$.  
Auction $M^2 = (\infty, \pc, \infty)$, with no license cap and with a price floor of $\pc$, will also sell all such licenses that have marginal value at least $\pc$.  We note, however, that $M^2$ additionally also includes the marginal contribution of the first $C_1$ licenses.  But we claim that this contribution is non-negative:
when the event occurs that more than $C_1$ licenses are sold, the marginal contribution of the first $C_1$ licenses to the welfare can only be greater than that of those beyond the first $C_1$.  Thus, since the expected welfare in the second part is non-negative, the welfare is only higher if we also include the contribution of the first $C_1$ licenses whenever more than $C_1$ licenses are sold. 
This is precisely the welfare of auction $M^2 = (\infty, \pc, \infty)$, so the welfare of $M^2$ is therefore at least that of the second of the two parts of the welfare of $M$.

We conclude that $W(M^1) + W(M^2) \geq W(M)$, and hence either $W(M^1)$ or $W(M^2)$ is at least $\frac{1}{2}W(M)$.
\end{proof}

\section{Omitted Proofs from Section~\ref{sec:main}}
\label{sec:proofs4}

\subsection{Proof of Lemma~\ref{lem.opt}} \label{apppf:lem.opt}

First recall the statement of the lemma.  If $C$ and $p$ are chosen to maximize $W(M(C,p))$, then 
$\E[W(V,x(V)) \ |\ d(V) \geq C] \geq 0.$

\begin{proof}
Suppose not.  Then it must be that $\E[W(V,x(V)) \ |\ d(V) \geq C] < 0.$  We will show this implies $W(M(C-1,p)) > W(M(C,p))$, contradicting the optimality of $C$.  To see why, note that when $d(p) < C$, the welfare of the two auctions is identical.  Write $x$ and $x'$ for the allocation functions from $M(C,p)$ and $M(C-1,p)$, respectively.  When $x \geq C$, we have $x' = C-1$, and $V(x') \geq \frac{C-1}{x} \cdot V(x)$ by concavity.  Similarly, $Q(x') \leq \frac{C-1}{x} \cdot Q(x)$ by convexity.  Thus, for any $V$ such that $d(V) \geq C$ and hence $x(V) = C$, we have
\[ V(x'(V)) - Q(x'(V)) \geq \frac{C-1}{x(V)} ( V(x(V)) - Q(x(V)) ) = \frac{C-1}{C} ( V(x(V)) - Q(x(V)) ). \]
Taking expectations, we therefore have
\begin{align*}
\E[V(x'(V)) - Q(x'(V)) \ |\ d(V) \geq C] &
\geq \frac{C-1}{C} ( \E[V(x(V)) - Q(x(V)) \ |\ d(V) \geq C] ) \\
& > \E[V(x(V)) - Q(x(V)) \ |\ d(V) \geq C]
\end{align*}
where the last inequality follows because $\E[V(x(V)) - Q(x(V)) \ |\ d(V) \geq C] < 0$ by assumption.  This implies $W(M(C-1,p)) > W(M(C,p))$, which is the desired contradiction.
\end{proof}

\subsection{Proof of Lemma~\ref{lem.unsafe.points}} \label{apppf:lem.unsafe.points}

First recall the statement of the lemma.  Choose some quantity $\tilde{C}$ and a number of licenses $x \leq \tilde{C}$.  
Then $P(\tilde{C}) \cdot x  - Q(x) \leq \tilde{C} \cdot ( P(\tilde{C}) - P(\tilde{C}/2) )$.


\begin{proof}
Suppose $x \geq \tilde{C}/2$.  Then $Q(x) \geq \frac{x}{\tilde{C}/2} \cdot Q(\tilde{C}/2)$ by convexity.  
So $P(\tilde{C}) \cdot x - Q(x) \leq x \cdot \left( P(\tilde{C}) - Q(\tilde{C}/2) / (\tilde{C}/2) \right) \leq \tilde{C} \cdot (P(\tilde{C}) - P(\tilde{C}/2) )$ as claimed.

Next suppose $x < \tilde{C}/2$.  Then $(x,P(\tilde{C}) \cdot x)$ and $(x,Q(x))$ both lie between the line through the origin with slope $P(\tilde{C})$, and the line between $(\tilde{C}, Q(\tilde{C}))$ and $(\tilde{C/2}, Q(\tilde{C}/2))$.  Their difference is therefore at most twice the difference between those two lines at x-coordinate $\tilde{C}/2$, which is $(\tilde{C}/2) \cdot ( P(\tilde{C}) - P(\tilde{C}/2) )$.
\end{proof}

\subsection{Omitted Details from the Proof of Theorem~\ref{thm:main}}\label{app:proof.thm.main}

Recall the statement of Theorem~\ref{thm:main}: for any cap $C$ and price floor $\pf$, there exists $C'$ and a constant $c$ such that $c \cdot W(M(C')) + W(M^{(1)}) \geq W(M(C,\pf))$.

As in the proof of Theorem~\ref{thm.q}, the welfare generated by $M(C,p)$ can be broken down as follows:

\begin{align*}
W(M(C,p)) &= \Pr_{V \sim F}[d(V) \geq C] \cdot \E[W(x(V),V) \ |\ d(V) \geq C] \\
&\quad + \Pr_{V \sim F}[d(V) < C] \cdot \E[W(x(V),V) \ |\ d(V) < C] \\
\end{align*}

We will break down the second term into three sub-terms.  Recall that we can assume $\Pr[ d(V) \geq C ] \leq 1 - 1/e$, as otherwise the result follows from Theorem~\ref{thm.q}.  Choose $\Cmed < C$ to be the largest integer such that $\Pr[ d(V) \geq \Cmed ] \geq 1 - 1/e$.  Note then that $\Pr[ d(V) \leq \Cmed ] \geq 1/e$.
We will write $x_i = x_i^* + x_i^{>} + x_i^{<}$.  If $x_i \geq \Cmed$ then we set $x_i^* = x_i$ and $x_i^{>} = x_i^{<} = 0$, otherwise we set $x_i^* = 0$.  In this latter case, we set $x_i^{>}$ and $x_i^{<}$ similarly to the proof of Theorem~\ref{thm.q}, except that we consider marginal values above and below $P(\Cmed)$ rather than $P(C)$.  That is, if $\theta_i$ is the largest $j \geq 1$ such that $v_i(j) \geq P(\Cmed)$ (or $\theta_i = 0$ if $v_i(1) < P(\Cmed)$), we have $x^{>}_i = \min\{x_i, \theta_i\}$ and $x^{<}_i = x_i - x^{>}_i$.  
And as before, we will define $V_i^{<}$ as $V_i^{<}(x) = V_i^{<}(x | \theta_i)$, 
the valuation of $i$ counting only those marginal values less than $P(\Cmed)$.

Using convexity of $Q$ as in case 1, we then have
\begin{align*}
W(M(C,p)) &\leq \Pr_{V \sim F}[d(V) \geq C] \cdot \E[W(x(V),V) \ |\ d(V) \geq C] \label{eq.2A}\tag{2A}\\
&\quad +  \int_{V \colon d(V) < C} (\sum_i V_i(x_i^{*}(V)) - Q(x^{*}(V))) dFV \label{eq.2B}\tag{2B}\\
&\quad +  \int_{V \colon d(V) < C} (\sum_i V_i(x_i^{>}(V)) - Q(x^{>}(V))) dFV \label{eq.2C}\tag{2C}\\
&\quad +  \int_{V \colon d(V) < C} (\sum_i V_i^{<}(x_i^{<}(V)) - Q(x^{<}(V)) ) dFV \label{eq.2D}\tag{2D}
\end{align*}
so that $W(M(C,p)) \leq (2A) + (2B) + (2C) + (2D)$.

As in Theorem~\ref{thm.q}, we have that $W(M(C)) \geq (2A)$.  For $(2B)$, note that the probability that every bidder $i$ has $x_i \leq \Cmed$ is at least $1/e$, from the definition of $\Cmed$.  If we write $p_i$ for the probability that $x_i \leq \Cmed$, then we have $\prod_i p_i \geq 1/e$.  Subject to this condition, the value of $\sum_i (1 - p_i)$ is maximized by setting all $p_i$ equal, in which case we have $\sum_i (1 - p_i) \leq n(1 - e^{-1/n}) \leq 1$.  But this means that the expected number of agents with $x_i > \Cmed$ is at most $1$.
Thus, $(2B)$ is at most the value of allocating to just a single bidder with $x_i > \Cmed$, which is at most the maximum possible welfare attainable from allocating to any one bidder.  We therefore have that $W(M^{(1)}) \geq (2B)$.  

We next claim that $W(M(\Cmed)) \geq (2C)/4$.  To see why, note that
\begin{align*}
& \int_{V \colon d(V) < C} \left(\sum_i V_i(x_i^{>}(V)) - Q(x^{>}(V)) \right) dFV \\
\leq\ & 2 \int_{V \colon d(V) < C} 1[x(V) < 2\Cmed] \cdot \left(\sum_i \left( V_i(x_i^{>}(V)) - x_i^{>} \cdot \tfrac{Q(x^{>}(V))}{x^{>}(V)}\right) \right) dFV \\
\leq\ & 4 \int_{V \colon d(V) < C} 1[x(V) < 2\Cmed] \cdot \left(\sum_i \left( \sum_{j \leq x_i^{>}} v_i(j) /2 \right) - x_i^{>} \cdot \tfrac{Q(\min\{\Cmed,x^{>}(V)\})}{\min\{\Cmed,x^{>}(V)\}}\right) dFV \\
\leq\ & 4 W(M(\Cmed))
\end{align*}
where the last inequality follows because, when imposing a cap of $\Cmed$ on an allocation of total size at most $2\Cmed$, the value obtained is at least the top half of all marginal values, which is more than half of all marginal values.

Our final step is to bound $(2D)$.  This is the most technical step, so let us first describe our approach.  This bound is similar to Claim~\ref{claim.1c} from the proof of Theorem~\ref{thm.q}, but with two additional complications.  First, we will handle the case where $\Cmed$ is odd, in which case $\Cmed/2$ is not an integer; we handle this by extending $V(\cdot)$ and $Q(\cdot)$ to the domain of non-negative real numbers via linear interpolation, studying a (fractional) cap of $\Cmed/2$ in that setting, then showing that one can round appropriately.  Second, in Claim~\ref{claim.1c} we made use of Lemma~\ref{lem.opt}, which said roughly that the expected welfare contribution of realizations where demand exceeds the cap must be non-negative.  Unfortunately, the same is not necessarily true of $\Cmed$: it could be that $\E_V[V(\Cmed) - Q(\Cmed)] < 0$ even after conditioning on the event that $d(V) \geq \Cmed$.  However, we will show that the optimality of $M(C,p)$ implies that this quantity cannot be \emph{too} negative, and this will allow us to show that $\E_V[V(\Cmed/2) - Q(\Cmed/2)]$ is positive and bounded away from $0$.  This will allow us to proceed as in Claim~\ref{claim.1c} and obtain the desired bound on $(2D)$.

Now for the details.  First, we extend $V(\cdot)$ and $Q(\cdot)$ to allow fractional inputs by linearly interpolating between integers, and likewise extend the definition of $P(x)$ to $Q(x)/x$ even for fractional $x$.  Note that $V$ is still concave and $Q$ is still convex under this extension.  The proof of Lemma~\ref{lem.unsafe.points} did not use integrality of $x$ or $\tilde{C}/2$, so it holds for these interpolated versions of $V$ and $Q$ as well.  Define $\Psi := (\Cmed/2) \cdot (P(\Cmed) - P(\Cmed/2))$.  Then Lemma~\ref{lem.unsafe.points} (with $\tilde{C} = \Cmed$), combined with the fact that the total demand is at most $\Cmed$ with probability at least $1/e$, implies that $(2D) \leq (2/e) \cdot \Psi$.

%
%

We next claim that for $\beta = \frac{8}{3(e-1)} \approx 1.55$, we must have $\Pr[V(\Cmed) < Q(\Cmed) - \beta \cdot \Psi \ |\ d(V) \geq \Cmed] \leq 1/4$.  That is, conditional on having total demand at least $\Cmed$, the probability that the welfare generated at $\Cmed$ is more negative than $-\beta \cdot \Psi$ is at most $1/4$.  Suppose not: then consider the difference in welfare between $M(C,p)$ and $M(C, P(\Cmed))$.  The difference is that (2D) would be removed, as would any (negative) contribution with $V(\Cmed) < Q(\Cmed)$ and $d(V) > \Cmed$.  As we noted above, the total contribution of (2D) to the welfare is at most $2 \Psi / e$.  But the contribution due to the event $V(\Cmed) < Q(\Cmed) - \Psi$ and $d(V) > \Cmed$ is at most $-(\beta \Psi)(1 - 1/e)(3/4)$.  So as long as $\beta \geq (8 / 3(e-1)) \approx 1.55$, we would have $W(M(C, P(\Cmed))) > W(M(C,p))$.  This contradicts the optimality of $M(C,p)$.

We can therefore assume $\Pr[V(\Cmed) \geq Q(\Cmed) - \beta \Psi \ |\ d(V) \geq \Cmed] \geq 1/4$.  By concavity of $V$ and convexity of $Q$, plus the fact that $Q(\Cmed) = 2(Q(\Cmed/2)+\Psi)$ from the definition of $\Psi$, we conclude that $\Pr[V(\Cmed/2) \geq Q(\Cmed/2) + (2-\beta) \Psi/2 \ |\ d(V) \geq \Cmed] \geq 1/4$.  

If $\Cmed$ is even, take $C' = \Cmed/2$.  Otherwise, since we defined $Q$ and $V$ by linear interpolation, the event that $V(\Cmed/2) \geq Q(\Cmed/2) + (2-\beta) \Psi/2$ implies that for some $C' \in \{ \lfloor \Cmed/2 \rfloor, \lceil \Cmed/2 \rceil \}$ we must have $V(C') \geq Q(C') + (2-\beta) \Psi/2$ as well.  Thus, in either case, there is an integer $C'$ such that $\Pr[V(C') \geq Q(C') + (2-\beta) \Psi/2 \ |\ d(V) \geq \Cmed] \geq 1/4$.  Since $d(V) \geq \Cmed$ with probability at least $1 - 1/e$, the total welfare generated by $M(C')$ from the events $d(V) \geq \Cmed$ is therefore at least $(2-\beta)(1-1/e) \Psi / 8$, and hence $W(M(C')) \geq (2-\beta)(1-1/e) \Psi / 8$.  Since $(2-\beta)(1-1/e)e/16 \geq 1/21$, the result follows.

We conclude that \[W(M(C)) + W(M^{(1)}) + 4 W(M(\Cmed)) + 21 W(M(\Cmed / 2)) \geq W(M(C, \pf)).\]  This implies Theorem~\ref{thm:main}, with $c = 26$.

\section{Cap-and-Price Auction Outcomes are not Fully Efficient}
\label{app:cap.price.bad}

We note that cap-and-price auctions cannot always implement the fully optimal allocation rule for every distribution $F$.  For an allocation rule to be implementable by some $M(C,\pf,\pc)$, it must be that on every input $\mathbf{V}$, either the total allocation is $C$, or the total allocation is at most $C$ and each agent wins precisely their marginal bids above $\pf$, or the total allocation is at least $C$ and each agent wins precisely their marginal bids above $\pc$.  When there is uncertainty about $\mathbf{V}$, such a restricted allocation rule might return suboptimal allocations on some realizations.  

In fact, we note that there are cases where, even under truthful reporting, no cap-and-price auction can achieve a non-vanishing approximation of the unrestricted welfare-optimal allocation. 

\begin{example} \label{ex:logscale}
	We present an example for which the expected welfare of any cap-and-price auction $M(C,\pf)$ is at most an $O(1/n)$-approximation to the unrestricted allocation that can optimize individually for each realization of the valuation curves, the welfare of which we call the ``first-best welfare'' in line with economics terminology.
	
		Let $Q(x) = x^2$.  There is a single firm participating in the auction.  That firm's valuation curve is drawn according to a distribution $F$ over $n$ different valuation curves, which we'll denote $V^{(1)}, \ldots, V^{(n)}$.  For all $i$, valuation $V^{(i)}$ is defined by $V^{(i)}(x) = 2^{i+1} \cdot x$.  These curves are depicted in Figure~\ref{fig:logscale}.  The probability that $V^{(i)}(\cdot)$ is drawn from $F$  is proportional to $\frac{1}{2^{2i}}$.  That is, the firm has valuation $V^{(i)}$ with probability $\frac{1}{\beta 2^{2i}}$, where $\beta = \sum _{i=1} ^n 2^{-2i} = \frac{1}{3} (1-2^{-2n})$ is the normalization constant.

Then in expectation over the realizations of $V(\cdot)$, the \emph{first-best} welfare is
$$\sum _{i=1} ^n 2^{2i} \frac{1}{\beta 2^{2i}} = \frac{1}{\beta} \cdot n.$$

Consider a cap-and-price auction $M(C, \pf)$.  For any cap $C$, the optimal price floor is to set $\pf = 2C$.  This eliminates all possible negative welfare contributions at $C$.  As a sanity check, any smaller price floor allows negative contributions, yet any larger price floor excludes positive contributions.  For $C' \in (2^k, 2^{k+1}]$, , if the firm has valuation $V^{(i)}$ with $i \leq k$, then all marginal values are strictly below $\pf$, so no licenses are allocated and the welfare generate is $0$.  On the other hand, if the firm has valuation $V^{(i)}$ for any $i \geq k +1$, then since at most $C < 2^{k+1}$ licenses can be purchased, the auction generates welfare at most $2^{i+1} (2^{k+1}) - (2^{k+1})^2 = 2^{2k+2}(2^{i-k} - 1)$.  This yields expected welfare
\begin{align*}
\sum _{i=k+1} ^n 2^{2k+2} (2^{i-k}-1) \cdot \frac{1}{\beta 2^{2i}} &= \frac{1}{\beta} \sum _{i=k+1} ^n  2^{-2(i-k-1)} (2^{i-k}-1)  \\
&= \frac{1}{\beta} \sum _{j=1} ^{n-k} 2^{-2(j-1)} (2^{j}-1) \\
&\leq \frac{1}{\beta} \sum _{j=1} ^{n-k} 2^{-(j-1)} \\
&= \frac{1}{\beta}(2 - 2^{-(n-k-1)}) \\
&\leq \frac{2}{\beta} = o(n).
\end{align*}
In comparison to $\opt$, any cap-and-price policy is off by an order of $n$, so no $o(1/n)$-approximation to the first-best welfare is possible.  This concludes the example.
\end{example}

\begin{figure} 
\centering
\includegraphics[scale=.5]{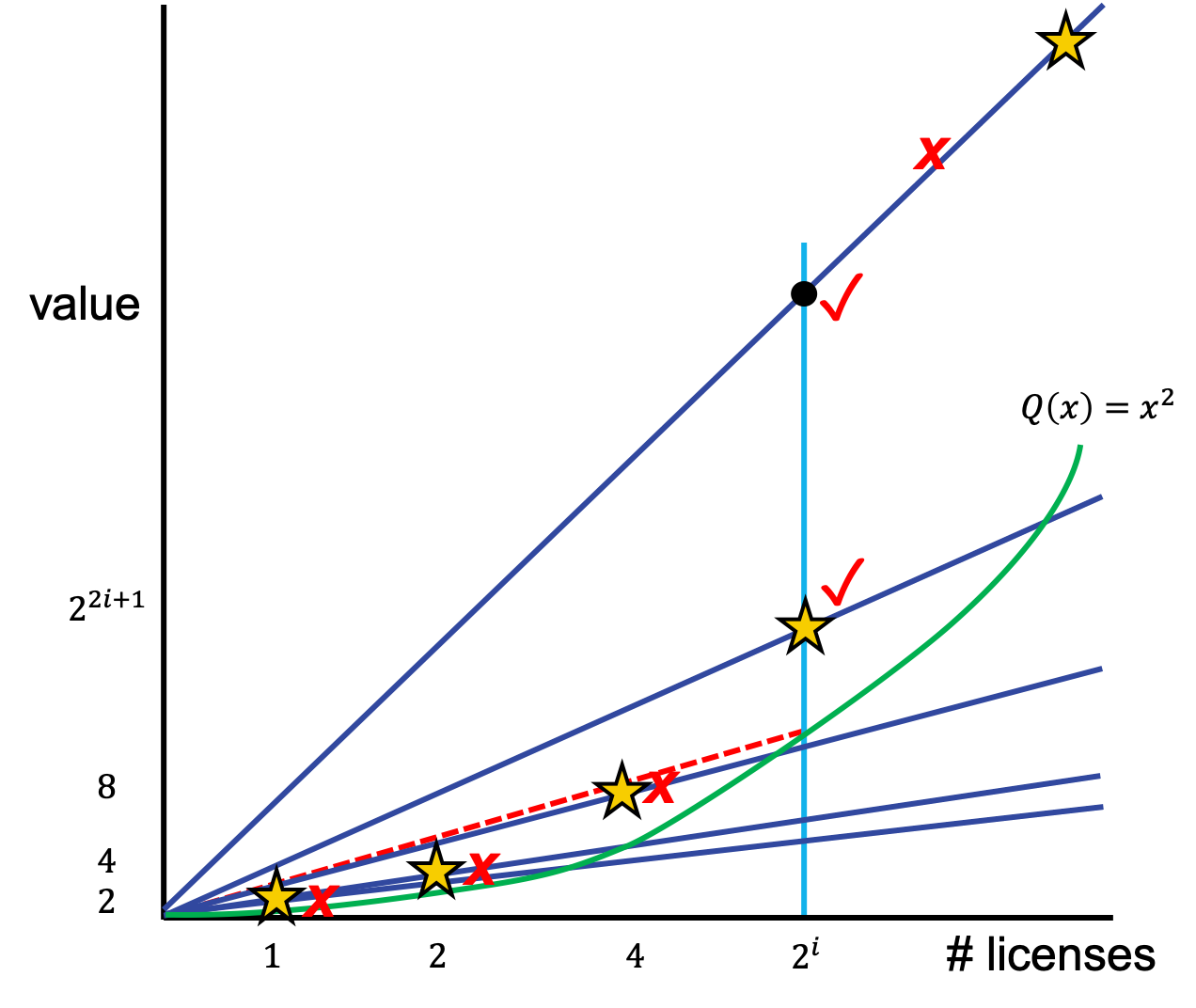}
\caption{A depiction of Example~\ref{ex:logscale}, with the welfare achieve by the vertical cap and dashed price floor denoted by the checks and x's.} \label{fig:logscale}
\end{figure}

\end{document}